\documentclass{article}
\usepackage{graphicx} 
\usepackage[a4paper,top=2cm,bottom=2cm,left=3cm,right=3cm,marginparwidth=1.75cm]{geometry}
\usepackage{subcaption}
\usepackage{xspace}
\usepackage{amssymb}
\usepackage{amsmath}
\usepackage{amsthm}
\usepackage{graphicx}
\usepackage{multicol}
\usepackage{diagbox}
\usepackage[colorlinks=true, allcolors=blue]{hyperref}
\usepackage{enumerate}    
\usepackage[ruled,linesnumbered,vlined]{algorithm2e}
\usepackage{enumitem}
\usepackage{tikz}
\usetikzlibrary{positioning,automata}
\usepackage{pifont}
\usepackage{algpseudocode}

\newtheorem{theorem}{Theorem}[section]
\newtheorem{lemma}[theorem]{Lemma}
\newtheorem{corollary}[theorem]{Corollary}

\newtheorem{definition}[theorem]{Definition}
\newtheorem{observation}[theorem]{Observation}
\newtheorem{example}[theorem]{Example}

\begin{document}

\title{Weighted Envy-Freeness in House Allocation}

\author{Sijia Dai 
 \thanks{Shenzhen Institute of Advanced Technology, Chinese Academy of Sciences.\{sj.dai, yc.xu, zhangyong\}@siat.ac.cn, xiaoshu712@126.com}
 \and 
 Yankai Chen $^*$
 \and
 Xiaowei Wu \thanks{IOTSC, University of Macau.  xiaoweiwu@um.edu.mo}
\and 
 Yicheng Xu $^*$ 
\and 
 Yong Zhang $^*$ 
}

\maketitle
%

\begin{abstract}
The classic house allocation problem involves assigning $m$ houses to $n$ agents based on their utility functions, ensuring each agent receives exactly one house. 
A key criterion in these problems is satisfying fairness constraints such as envy-freeness. 
We extend this problem by considering agents with arbitrary weights, focusing on the concept of weighted envy-freeness, which has been extensively studied in fair division.
We present a polynomial-time algorithm to determine whether weighted envy-free allocations exist and, if so, to compute one. 
Since weighted envy-free allocations do not always exist, we also investigate the potential of achieving such allocations through the use of subsidies.
We provide several characterizations for weighted envy-freeable allocations (allocations that can be turned weighted envy-free by introducing subsidies) and show that they do not always exist, which is different from the unweighted setting.
Furthermore, we explore the existence of weighted envy-freeable allocations in specific scenarios and outline the conditions under which they exist.

\end{abstract}
%
%
%
\section{Introduction}

The house allocation problem is a significant class of resource allocation issues, both in theory and practical application. 
These problems involve distributing a set of indivisible items, such as houses, offices, or tasks, among a group of agents with possibly differing preferences, under the constraint that each agent is assigned exactly one item~\cite{hylland1979,zhou1990,abdulkadirouglu1998}. Various desirable properties were proposed for house allocation problems, or more broadly speaking resource allocation problems.
One of the critical challenges is achieving fairness and efficiency without the use of monetary exchange, which is often the case in subsidized housing scenarios like on-campus housing for students or public housing for low-income families. 

The house allocation problem has been primarily explored for designing incentive-compatible mechanisms and ensuring economic efficiency~\cite{abdulkadirouglu2003,svensson1999}.
Recent works in house allocation have increasingly focused on fairness, often captured by well-studied notions such as envy-freeness (EF), which requires that every agent is assigned one of his most preferred houses among all the allocated houses~\cite{varian1973,foley1966resource,Haris2024,beynier2019local}.
Unfortunately, envy-free allocations do not always exist. For example, consider two agents who both prefer the same house.
Gan et al.~\cite{gan2019envy} developed a polynomial-time algorithm to determine the existence of envy-free house allocations and to compute one if they exist. When no envy-free allocation exists, the typical aim is to compromise on the fairness objective.
Kamiyama et al.~\cite{kamiyama2021complexity} demonstrated that finding allocations which minimize the number of envious agents is NP-complete, even for binary utilities. 
Madathil et al.~\cite{madathil2024cost} proposed three fairness metrics for the amount of envy experienced by an agent and studied minimization problems corresponding to each of these metrics.
Hosseini et al.~\cite{hosseini2023graphical} considered scenarios in which agents have identical utility functions and only envy others based on a predefined social graph. Their goal was to minimize the total envy among all agents, and they also completed a characterization of approximability across various graphs in~\cite{hosseini2024tight}.
A very recent study by Hosseini et al.~\cite{hosseini2024degree} investigated the trade-off between efficiency and fairness in house allocation. 

In order to get around the possible non-existence of envy-free allocations, recent works turned to allocations with subsidies, which have been well-studied in fair division problems~\cite{conf/wine/WuZZ23,conf/sigecom/BrustleDNSV20,SAGT2019}. 
The aforementioned result of Halpern and Shah~\cite{SAGT2019} implies that every house allocation problem instance can always achieve envy-freeness by providing subsidy.
However, Choo et al.~\cite{CHOO2024} demonstrated that finding the minimum subsidy required to achieve envy-freeness is NP-hard in general.
Another natural direction is to investigate the effects of strategic behavior in fair house allocation.
In this vein, recent works investigated the trade-offs between competing notions of efficiency, fairness and incentives in allocation mechanisms~\cite{shende2023strategy,krysta2014size,thakral2016public}.

To the best of our knowledge, the house allocation problem has not been studied in the context of weighted envy-freeness, which has been studied widely in fair division problems~\cite{brams1996fair,robertson1998cake,chakraborty2021weighted,wu2023weighted,aziz2023best,springer2024almost}.
Envy-freeness can be extended to a more general setting where agents have weights that designate their entitlements.
The weights of agents can represent other publicly recognized and accepted measures of entitlement, such as eligibility or merit. 
A common example in housing allocation is that larger families naturally require larger living spaces compared to smaller families. 
Therefore, even if the larger family gets a more valuable house than the smaller family, envy could still occur: the larger family naturally deserves more.

In this article, we initiate the study of weighted envy-freeness for the house allocation problem.
We consider the existence of weighted envy-free allocations, with and without subsidies, which have not been considered for the house allocation problem.
It remains unknown whether weighted envy-free allocations exist when subsidies are allowed.

\subsection{Our Results}

Our work considers agents with arbitrary weights in house allocation problems and focuses on the fairness notion of weighted envy-freeness.
We present the results regarding the existence and computation of weighted envy-free allocations, and weighted envy-free allocations with subsidy.

\smallskip
\noindent
\textbf{Result 1.}
We present a polynomial-time algorithm that determines the existence of weighted envy-free allocations in house allocation and computes such an allocation, if they exist.
\smallskip

In our algorithm, we preprocess agents' utility functions according to their weights before proceeding with the allocation. Subsequently, we employ a matching algorithm to evaluate potential allocation configurations that would confirm the existence of weighted envy-free allocations in the instance.
We iteratively ban some agents from receiving some houses (unless a weighted envy-free allocation is found) and show that it will not change the existence of weighted envy-free allocations.
If at some point the allowed assignments do not admit weighted envy-free allocations, our algorithm stops and announces that no weighted envy-free allocations exist.

\smallskip
\noindent
\textbf{Result 2.}
We characterize weighted envy-freeable allocations, showing sufficient and necessary conditions under which an allocation is weighted envy-freeable.
In contrast to the unweighted case, we give an instance showing that weighted envy-freeable allocations do not always exist.
\smallskip

We refer to an allocation of houses as weighted envy-freeable if it is possible to eliminate envy by paying each agent some amount of money.
Surprisingly, weighted envy-freeable allocations do not always exist in the house allocation problem when agents have arbitrary weights.
Interestingly, our characterization reveals that weighted envy-freeable fair allocations (in which all items should be allocated and there is no constraint on the number of items each agent can receive) always exist.
In other words, the agent weights and allocation constraints (together) disallow the existence of weighted envy-free allocations, even if we are allowed to (arbitrarily) subsidize the agents.

\smallskip
\noindent
\textbf{Result 3.}
We discuss the existence of weighted envy-freeable allocations and propose polynomial time algorithms to compute them in special cases, including those with identical utility functions, two types of agents, and bi-valued utility functions.

\smallskip

Our results address the issue of achieving weighted envy-freeness for agents with arbitrary weights in house allocation problems, which has been overlooked in existing  works.

\subsection{Other Related Works}

In fair division, deciding the existence of envy-free allocations with zero subsidy is already NP-hard. 
In contrast, this task can be efficiently solved for house allocation~\cite{gan2019envy}. 
%
The use of subsidies has been well-explored in the fair allocation of indivisible goods~\cite{conf/sigecom/BrustleDNSV20,goko2024fair,SAGT2019,barman2022achieving}.
Another closely related work pertains to the rent division problem, which involves allocating $m = n$ indivisible goods among $n$ agents and dividing the fixed rent among them~\cite{gal2017fairest,aragones1995derivation,edward1999rental,svensson1983large}. 
More general models where envy-freeness is achieved through subsidies have also been considered in~\cite{haake2002bidding,meertens2002envy}.
When considering agents have arbitrary weights, our setting is similar to the fair allocation under weighted setting.
Other related works explored relaxed weighted fairness notions with subsidies, aiming to minimize the required subsidy~\cite{conf/wine/WuZZ23,wu2024tree}.
For a comprehensive review of existing research on weighted and unweighted fair allocation, please refer to the recent surveys~\cite{suksompong2024weighted,amanatidis2023fair}.

\subsection{Organization of the Paper}
The following sections of this paper are structured as follows.
We first provide the notations and definitions for weighted envy-freeness of house allocation in Section 2. 
In Section 3, We provide a polynomial-time algorithm to determine the existence of weighted envy-free house allocations and to compute one if they exist.
In Section 4, we consider weighted envy-free house allocations with subsidies and present the non-existence results of allocations that are weighted envy-freeable.
We discuss some special cases regarding the existence of weighted envy-freeable allocations in Section 5, and conclude the paper with the open questions in Section 6.


\section{Preliminaries}
In the following, we introduce the notations and fairness notions for the house allocation problem. 
Let $N=\{1,2,\dots,n\}$ denote the set of $n$ agents and $H=\left\{h_{1}, \dots, h_{m}\right\}$ denote the set of $m$ houses, where $m \ge n$.
Each agent $i\in N$ has a non-negative utility $v_i(h)$ for each house $h \in H$. 
We denote the vector of utility functions by $\mathbf{v}=(v_1,\dots,v_n)$.
An allocation $\mathbf{A} =(A_1,\dots, A_n)$ is a list of $n$ distinct houses, 
where $A_i \in H$ is the house allocated to agent $i$. 
It would sometimes be helpful to regard an allocation as a $N$-saturating matching\footnote{A matching is called $N$-saturating if all agents in $N$ are matched.} in a complete bipartite graph, where the two sides of vertices are the agents $N$ and the houses $H$.

In the weighted setting, each agent $i\in N$ is additionally endowed with a weight $w_i>0$ representing her entitlement.
Let $W=(w_1,\dots,w_n)$.
The unweighted setting corresponds to the special case where all weights are equal.
We define a house allocation problem to be the tuple $\mathcal{I}=(N,H,W,\mathbf{v})$.
The following fairness notion is central to our work.

\begin{definition}[WEF Allocation]
\label{definition:WEF Allocation}
In the weighted setting, an allocation $\mathbf{A}$ is  weighted envy-free (WEF) if for every pair of agents $i,j\in N$, it holds that $$\frac{v_i(A_i)}{w_i} \ge \frac{v_i(A_j)}{w_j}.$$
\end{definition}

When weighted envy-free allocations do not exist, we aim to achieve a weighted envy-free allocation through subsidies.
We denote by $p_i \in \mathbb{R}_{\ge 0}$ the amount of money received by agent $i$.
The subsidies provided to the agents are represented by a non-negative \textit{subsidy vector} $\mathbf{P}=(p_1,p_2,\dots,p_n)$.
We refer to an allocation $\mathbf{A}$ with a corresponding subsidy vector $\mathbf{P}$ as an outcome $(\mathbf{A},\mathbf{P})$.
Under the outcome $(\mathbf{A},\mathbf{P})$, the utility of agent $i$ is $v_i(A_i)+p_i$.
Note that allocation $\mathbf{A}$ is equivalent to allocation with subsidies $(\mathbf{A},\mathbf{0})$, where each agent receives zero subsidy. 
We can now introduce the definition of weighted envy-freeness outcome as follows.

\begin{definition}[WEF Outcome]
\label{definition:WEF Outcome}
An outcome $(\mathbf{A},\mathbf{P})$ is weighted envy-free if for every pair of agents $i,j\in N$, it holds that 
$$\frac{v_i(A_i)+p_i}{w_i} \ge \frac{v_i(A_j)+p_j}{w_j}.$$
\end{definition}

We say that a subsidy vector $\mathbf{P}$ is \textit{envy-eliminating} for an allocation $\mathbf{A}$ if the outcome $(\mathbf{A}, \mathbf{P})$ satisfies weighted envy-freeness. 
Let $\mathcal{P}(\mathbf{A})$ be the set of envy-eliminating subsidy vector for $\mathbf{A}$.

\begin{definition}[WEFable]
\label{definition:WEFable}
An allocation  $\mathbf{A}$ is called weighted envy-freeable (WEFable) if there exists a subsidy vector $\mathbf{P}$ such that the outcome $(\mathbf{A}, \mathbf{P})$ is weighted envy-free, that is, $\mathcal{P}(\mathbf{A}) \neq \emptyset$.
\end{definition}

A result by Halpern and Shah~\cite{SAGT2019} implies the following observation directly.
\begin{observation}
\label{observation:EFable exist}
For the unweighted case, computing a maximum weight matching between $N$ and $H$ (or any $n$ houses) gives an envy-freeable allocation in house allocation problems.
It means that every instance of the house allocation problem (with arbitrary n and m) has at least one envy-freeable allocation. 
\end{observation}

\begin{definition}[PO]
\label{definition:PO}
An allocation $\mathbf{A}$ is Pareto optimal (PO) if there exists no allocation $\mathbf{A'}$ 
such that $v_i(A_i') \ge v_i(A_i)$ for all  $i \in N$ and $v_j(A_j') > v_j(A_j)$ for some $j \in N$.
\end{definition}

\section{Computation of Weighted Envy-Free Allocations}
\label{section:Computation of Weighted Envy-Free Allocations}
In this section, we propose a polynomial time algorithm that given any instance of the weighted house allocation problem, decides whether WEF allocations exist and, if they do, computes one.
Our algorithm is inspired by the algorithm proposed by Gan et al.~\cite{gan2019envy} for the unweighted version of the problem, and can be regarded as a generalization of their algorithm.

\paragraph{General Idea.}
We employ a filtering algorithm to preliminarily screen out infeasible assignments.
For example, suppose in the unweighted case there are two agents both value the same house as their only favourite one.
Then such house cannot be allocated to any agent as otherwise envy occurs.
Therefore we can safely remove the house from the instance without changing the existence of envy-free allocations.
However, when agents have arbitrary weights, it may still be feasible to allocate the house to an agent with much larger weights, in an WEF allocation.
Therefore, instead of removing the house, we can only ``ban'' the house from being allocated to certain agents.
Note that an agent might be banned from receiving a house by any other agent, as long as the later believes that such allocation will result in inevitable weighted envy. 
We formalizing this idea into an algorithm that iteratively bans some agents from getting some houses, which corresponds to deleting edges in a bipartite graph between agents and houses, until all remaining edges are feasible.
If the resulting bipartite graph admits a matching that matches all the agents, a WEF allocation is returned; otherwise more edges will be deleted and the algorithm recurses.
If no WEF allocation is returned when some agent has degree $0$ in the bipartite graph, the algorithm terminates and announces that no WEF allocations exist.

\paragraph{The Algorithm.}
To formalize the process, we define the following \emph{virtual instance}.
Given the instance $\mathcal{I}=(N,H,W,\mathbf{v})$, we create a set of virtual houses as
\begin{equation*}
    E = \{ e(i,h) : i\in N, h\in H \}.
\end{equation*}
Note that each virtual house $e(i,h)$ corresponds to an assignment of house $h$ to agent $i$\footnote{To avoid confusion, in this paper we use ``assignment'' to refer to assigning a house to an agent; and ``allocation'' to refer to the overall allocation (of houses to agents).} and we have $|E| = n\cdot m$.
We can also interpret $E$ as the set of edges between agents $N$ and houses $H$ in a bipartite graph.
We expand the utility function of each agent $i$ to a utility function $v'_i$ on the virtual houses, by defining $v'_i(e(j,h)) = \frac{v_i(h)}{w_j}$ for all $e(j,h)\in E$.
We further define $\mathsf{Top}_i(E) = \arg\max_{e\in E}\{ v'_i(e) \}$ as the subroutine that returns the set of virtual houses in $E$ that agent $i$ values the highest.
Recall that in a WEF allocation $\mathbf{A}$, agent $i$ is weighted envy-free towards agent $j$ if $\frac{v_i(A_i)}{w_i} \geq \frac{v_i(A_j)}{w_j}$, which translates into $v'_i(e(i,A_i)) \geq v'_i(e(j,A_j))$.
Therefore, suppose $\mathsf{Top}_i(E)$ does not contain any virtual house of the form $e(i,\cdot)$, then all assignments correspond to the virtual houses in $\mathsf{Top}_i(E)$ cannot appear in any WEF allocation, as otherwise the allocation is not WEF for agent $i$, no matter which house agent $i$ gets.
Formally speaking, we let $E_i = \{ e(i,h):h\in H \}$ for all $i\in N$.
As long as there exists an agent $i\in N$ such that $\mathsf{Top}_i(E) \cap E_i = \emptyset$, we update $E \gets E\setminus \mathsf{Top}_i(E)$ (we show in Lemma~\ref{lemma:safe-removal-of-edges} that such operation does not change the existence of WEF allocations).
The above procedure repeats until $|E| < n$ (in this case we can guarantee that no WEF allocations exist) or for all $i\in N$, $\mathsf{Top}_i(E)$ contains a virtual house $e(i,h)$ for some $h$.
Note that if agent $i$ is assigned house $h$ and $e(i,h)\in \mathsf{Top}_i(E)$, then as long as all other assignments are contained in $E$, the allocation is WEF for agent $i$ (see Lemma~\ref{lemma:returned-allocation-is-WEF} for a formal proof).
However, it is possible that there exists another agent $j$ whose favourite assignment is $e(j,h)$, in which case we need to decide which agent the house $h$ should be assigned to.
In this case, we use matching theory to check whether there exist allocations that assign each agent one of her favourite virtual houses; and if none is found, we use a subroutine by Gan et al.~\cite{gan2019envy} to find minimal Hall violators and remove more edges.

We summarize the details of the algorithm in Algorithm~\ref{Alg1}.

\begin{definition}[Hall Violator]
For a bipartite graph $G=(N,H,E')$, a Hall violator $Z\subseteq N$ is a set of vertices such that $|Z|> |H(Z)|$, where $H(Z)\subseteq Y$ are the vertices adjacent to $Z$. 
It is a minimal Hall violator if no subset of $Z$ is a Hall violator.
\end{definition}

\begin{algorithm}[htp]
\caption{Algorithm for Computing Weighted Envy-Free Allocations}

\label{Alg1}
\KwIn{An instance $\mathcal{I}=(N,H,W,\mathbf{v})$;}
\KwOut{A WEF allocation $\mathbf{A}$ or ``No such allocation'';}
Define $E = \{ e(i,h): i\in N, h\in H \}$ and $E_i = \{ e(i,h): h\in H \}$ for all $i\in N$ \;
Let $v'_i(e(j,h)) = \frac{v_i(h)}{w_j}$ for all $i\in N$ and $e(j,h)\in E$ ;\\
\While{$|E| \geq n$}
{
\While{$\exists i\in N, \mathsf{Top}_i(E)\cap E_i = \emptyset$}
{
Update $E \gets E \setminus  \mathsf{Top}_i(E)$ \;
}
Construct a bipartite graph $G = (N,H,E')$ where there is an edge from agent $i$ to house $h$
if and only if $e(i,h) \in \mathsf{Top}_i(E)$ \;
\If{there exists an $N$-saturating matching in $G$}
{
\textbf{return} the corresponding allocation $\mathbf{A}$ and \textbf{terminate};
}
\Else
{
Find a minimal Hall violator $Z \subseteq N$ \;
Update $E \gets E \setminus \{ e(i,h): i\in Z, (i,h) \in E' \}$ \;
}
}
\textbf{output} ``No such allocation''.
\end{algorithm}

Next, we establish some useful lemmas that will be used to prove the correctness of our algorithm.
We first show that if the algorithm returns an allocation (in line 8), then it is WEF.

\begin{lemma} \label{lemma:returned-allocation-is-WEF}
If Algorithm~\ref{Alg1} returns an allocation $\mathbf{A}$, then the allocation is WEF.
\end{lemma}
\begin{proof}
    By the design of the algorithm, for all $i\in N$ we have $(i,A_i) \in E'$, which implies that $e(i,A_i) \in \mathsf{Top}_i(E)$, where $E'$ and $E$ refer to the set of edges in the bipartite graph and the set of feasible assignments, right before the allocation $\mathbf{A}$ is returned.
    For any other agent $j\neq i$, since $e(j,A_j) \in E$, we have $v'_i(e(i,A_i)) \geq v'_i(e(j,A_j))$, which implies $\frac{v_i(A_i)}{w_i} \geq \frac{v_i(A_j)}{w_j}$.
    Therefore the allocation is WEF.
\end{proof}

Next, we show that if a house $h$ is the (weighted) favourite for both agents $i$ and $j$, then they must have the same weight.

\begin{lemma}\label{lemma:same-favourite-same-weight}
Suppose at certain point in time we have $e(i,h)\in \mathsf{Top}_i(E)$ and $e(j,h)\in \mathsf{Top}_j(E)$ for two agents $i,j\in N$ and the same house $h$, then $w_i=w_j$.
\end{lemma}
\begin{proof}
Suppose otherwise, and we assume that $w_i<w_j$.
Then, we have 
\begin{equation*}
    v'_j(e(i,h)) = \frac{v_j(h)}{w_i}>\frac{v_j(h)}{w_j} = v'_j(e(j,h)),
\end{equation*}
which contradicts with $e(i,h)\in E$ and $e(j,h)\in \mathsf{Top}_j(E)$.
\end{proof}

The next lemma is crucial for showing the correctness of the algorithm, which states that all edges removed from $E$ cannot appear in any WEF allocation.

\begin{lemma}\label{lemma:safe-removal-of-edges}
At any point in time, for all $e(i,h) \notin E$, agent $i$ cannot be assigned house $h$ in any WEF allocation.
In other words, for every WEF allocation (if there is any), all $n$ assignments in the allocation appear in $E$.
\end{lemma}
\begin{proof}
We prove the lemma by contradiction: assume that there exists a WEF allocation $\mathbf{A}$ and $e(i,A_i)$ is removed from $E$ by our algorithm.
We further assume that $e(i,A_i)$ is the first assignment in $\mathbf{A}$ that is removed.
Note that Algorithm~\ref{Alg1} removes assignments under two circumstances: in line 5 and line 11. We will prove the lemma in both cases. 

Suppose that $e(i,A_i)$ is removed due to the operation in line 5, i.e., there exists $j\neq i$ such that $\mathsf{Top}_j(E)\cap E_j = \emptyset$ and $e(i,A_i)\in \mathsf{Top}_j(E)$.
Since $e(i,A_i)$ is the first assignment in $\mathbf{A}$ that is removed, we have $e(j,A_j)\in E$ before $e(i,A_i)$ is removed.
Since $e(i,A_i)\in \mathsf{Top}_j(E)$ but $e(j,A_j)\notin \mathsf{Top}_j(E)$, we have
\begin{equation*}
    \frac{v_j(A_i)}{w_j} = v'_j(e(j,A_j)) < v'_j(e(i,A_i)) = \frac{v_j(A_i)}{w_i},
\end{equation*}
which contradicts with $\mathbf{A}$ being WEF.

Next we consider the case when $e(i,A_i)$ is removed due to the operation in line 11, i.e., $i$ is contained in some minimal Hall violator $Z$, and $(i,A_i)\in E'$.
Let $X = \{ j\in Z: (j,A_j)\in E' \}$.
By definition we have $i\in X\subseteq Z$.
Moreover we have $Z \setminus X \neq \emptyset$ because otherwise 
\begin{equation*}
    |H(Z)| = |H(X)| \geq |X| = |Z|,
\end{equation*}
contradicting with $Z$ being a Hall violator.
Moreover, we know that $Z\setminus X$ is not a Hall violator, as $Z$ is a minimal Hall violator.
Therefore we have $H(Z\setminus X) \geq |Z \setminus X|$.
Together with the fact that each agent $j\in X$ has a unique neighbor $A_j\in H(X)$, we conclude that there must exist $l\in Z\setminus X$ and $j\in X$ such that $(l,A_j)\in E'$, which implies that $e(l, A_j)\in \mathsf{Top}_l(E)$.
Since we also have $e(j, A_j)\in \mathsf{Top}_j(E)$, Lemma~\ref{lemma:same-favourite-same-weight} implies that $w_j = w_l$.
Since $e(i,A_i)$ is the first assignment that is removed from $E$, we have $e(l,A_l)\in E$ at the moment.
Moreover, since $l\notin X$, by definition of $X$ we have $(l,A_l)\notin E'$, i.e., $e(l,A_l)\notin \mathsf{Top}_l(E)$, which implies
\begin{equation*}
    \frac{v_l(A_l)}{w_l} = v'_l(e(l,A_l)) < v'_l(e(l,A_j)) = \frac{v_l(A_j)}{w_l}  = \frac{v_l(A_j)}{w_j},
\end{equation*}
which contradicts with $\mathbf{A}$ being WEF.
\end{proof}

Finally, we prove the main result of this section.

\begin{theorem} \label{theorem:decide-WEF}
Algorithm~\ref{Alg1} is a polynomial time algorithm that decides whether weighted envy-free allocations exist and, if so, computes one such allocation.
\end{theorem}
\begin{proof}
We first prove the correctness of our algorithm, then show that it runs in polynomial time.

By Lemma~\ref{lemma:returned-allocation-is-WEF}, if some allocation is returned by the algorithm, then it is WEF. 
Thus it suffices to show that if the algorithm terminates with no allocation returned, then the instance does not admit any WEF allocation.
Assume for contradiction that there exists a WEF allocation $\mathbf{A}$.
By Lemma~\ref{lemma:safe-removal-of-edges}, for all $i\in N$, we have $e(i,A_i) \in E$ throughout the whole algorithm, which implies $|E| \geq n$.
Therefore the algorithm will not output ``No such allocation'', which is a contradiction.

Next we show that the algorithm runs in $O(n^3 m^2)$ time.
Note that in each while-loop, if the algorithm does not terminate, at least one assignment will be removed from $E$.
We show that if $k$ assignments are removed from $E$ in some while-loop, then the while-loop executes in $O(kn^2m)$ time.
Since $|E| = O(nm)$ and the algorithm terminates when $|E|<n$, the above argument implies that the algorithm runs in $O(n^3 m^2)$ time.

We first focus on lines 4 - 5.
For each agent $i\in N$, it takes $O(|E|)$ time to check whether $\mathsf{Top}_i(E) \cap E_i = \emptyset$, by scanning $E$ and recording the maximum value of assignments in $E\cap E_i$ and that in $E\setminus E_i$.
Therefore in $O(n\cdot |E|) = O(n^2 m)$ time we can either break the while-loop in line 4, or remove at least one assignment.
Thus lines 4 - 5 end in $O(k n^2 m)$ time.
Constructing the bipartite graph in line 6 takes $O(n^2 m)$ time (by identifying $\mathsf{Top}_i(E)\cap E_i$ for all $i\in N$).
Computing the maximum (cardinality) matching takes $O(n\cdot |E'|) = O(n^2 m)$ time.
If the maximum matching is not $N$-saturating, we can find a minimal Hall violator in $O(n\cdot |E'|) = O(n^2 m)$ time by constructing the residual graph (see Lemma 2.1 of~\cite{gan2019envy} for the details), after which we remove more assignments.
\end{proof}

We further show that if an allocation is returned then it must be Pareto optimal among all WEF allocations.

\begin{corollary}
The allocation returned by Algorithm~\ref{Alg1} is Pareto optimal among all weighted envy-free allocations.
\end{corollary}
\begin{proof}
For the weighted house allocation problem $\mathcal{I}=(N,H,W,\mathbf{v})$, suppose it admits weighted envy-free allocations, and Algorithm~\ref{Alg1} returns an allocation $\mathbf{A}$. 
Assume for contradiction that there exists another allocation $\mathbf{B}$, such that $v_i(B_i)>v_i(A_i)$ for some $i\in N$ and $v_j(B_j)\geq v_j(A_j)$ for all $j \in N\setminus \{i\}$.
By Lemma~\ref{lemma:safe-removal-of-edges}, all assignments in $\mathbf{A}$ and $\mathbf{B}$ are preserved in $E$.
Since the assignment $e(i,B_i)$ ranks before $e(i,A_i)$ in $E$ under $v'_i$, we have $(i,A_i)\notin E'$, which contradicts the fact that $A_i$ is assigned to agent $i$.
\end{proof}


\section{Weighted Envy-Free Allocation with Subsidy}

The aforementioned result of Section~\ref{section:Computation of Weighted Envy-Free Allocations} implies that we can decide in polynomial time whether there exist WEF allocations with zero subsidy. 
When WEF allocations do not exist, it may still be possible to achieve weighted envy-freeness if a third party is willing to subsidize the envious agents.
As observed in existing works~\cite{SAGT2019,conf/sigecom/BrustleDNSV20}, not every allocation can be turned envy-free by introducing subsidies, i.e., is envy-freeable.
However, for the unweighted case we can always find allocations that are envy-freeable (see Observation~\ref{observation:EFable exist}).
In this section, our goal is to characterize WEFable allocations and, given a WEFable allocation, to find the minimum amount of envy-eliminating subsidy vector.
One of our main discovery is that WEFable allocations do not always exist, even in some very special cases.

\subsection{Weighted Envy Graph}
\label{section:Weighted Envy Graph}

Envy graph, proposed by Halpern and Shah~\cite{SAGT2019} for the unweighted version of the problem, has been an essential tool for studying envy-freeable allocations. 
In this section, we introduce the weighted envy graph and use it to characterize WEFable allocations, generalizing the results of Halpern and Shah~\cite{SAGT2019} to weighted allocation problems.

\paragraph{Weighted Envy Graph.}
Given any allocation $\mathbf{A} = (A_1,A_2,\ldots,A_n)$, we construct a weighted directed complete graph $G_\mathbf{A}$ in which the directed edge $(i,j)$ has weight
\begin{equation*}
    w(i,j) = \frac{v_i(A_j)}{w_j} - \frac{v_i(A_i)}{w_i}.
\end{equation*}
Note that $w(i,j)$ is the amount of envy agent $i$ has for agent $j$, which can be negative if agent $i$ does not envy agent $j$ for having some subsidy.
We further define a path $\mathcal{L}$ in $G_\mathbf{A}$ as a sequence of nodes $(i_0,i_1,i_2, \dots, i_k)$.
The path is a \textit{cycle} if $i_0=i_k$.  
Note that by definition, $w(i,i)=0$ for each $i \in N$.
Let $w(\mathcal{L})$ denote the sum of the weight of edges in the path:
\begin{equation*}
    w(\mathcal{L})=\sum_{t=0}^{k-1}w(i_t,i_{t+1}).
\end{equation*}
Given any pair of agents $i,j \in N$, let $\ell(i,j)$ be the maximum weight of any path that starts at $i$ and ends at $j$.
Let $\ell(i)=\max_{j \in N}\ell(i,j)$  be the maximum weight of any path starting at $i$.

\subsection{Characterizing WEFable Allocations}
\label{section:Weighted-Envy-Freeable-Allocations}

In this section, our goal is to characterize WEFable allocations and, given a weighted envy-freeable allocation, to find the minimum amount of envy-eliminating subsidy vector.
We remark that our characterization holds for any allocation (not necessarily house allocation), in which an agent can receive any number of houses.

\begin{lemma}
\label{lemma:no-positive-weight-cycles}
An allocation $\mathbf{A}$ is WEFable if and only if the corresponding weighted envy graph $G_{\mathbf{A}}$ has no positive-weight cycles.
\end{lemma}
\begin{proof}
We first prove the ``if'' direction, which states that an allocation $\mathbf{A}$ is WEFable if the corresponding weighted envy graph $G_{\mathbf{A}}$ has no positive-weight cycles.
Suppose $G_{\mathbf{A}}$ has no positive-weight cycles.
Recall that we define $\ell(i)$ to be the maximum weight of any path starting from $i$ in $G_{\mathbf{A}}$.
Let $p_i = w_i \cdot \ell(i) \geq 0$ for each $i \in N$, resulting in a subsidy vector $\mathbf{P}$.
By the definition of $w(\mathcal{L})$ and $\ell(i)$, for any two agents $i,j \in N$, we have
   \begin{equation*}
        \frac{p_i}{w_i} = \ell(i) \geq w(i,j) + \ell(j) = \frac{v_i(A_j)}{w_j} - \frac{v_i(A_i)}{w_i} + \frac{p_j}{w_j},
    \end{equation*}
       where the first inequality holds since edge $(i,j)$ together with a path starting from agent $j$ gives a path starting from agent $i$.
  Reordering the above equation implies that agent $i$ is weighted envy-free towards agent $j$:
    \begin{equation*}
         \frac{v_i(A_i) + p_i}{w_i} \geq \frac{v_i(A_j) + p_j}{w_j}.
    \end{equation*}     
It means that the outcome $(\mathbf{A}, \mathbf{P})$ is WEF.
Thus, the allocation $\mathbf{A}$ is WEFable.

\smallskip

Next we prove the ``only if'' direction, which states that if an allocation $\mathbf{A}$ is WEFable then the corresponding weighted envy graph $G_{\mathbf{A}}$ has no positive-weight cycles.
Fix any cycle $\mathcal{C} = (i_0,i_1,\ldots,i_{k})$, where $i_0 = i_{k}$, by weighted envy-freeness of $(\mathbf{A},\mathbf{P})$, for any $i_t$ and $i_{t+1}$, where $t \in \{0, 1,\ldots,k-1\}$, we have
    \begin{equation*}
        \frac{v_{i_t}(A_{i_t}) + p_{i_t}}{w_{i_t}} \geq \frac{v_{i_t}(A_{i_{t+1}}) + p_{i_{t+1}}}{w_{i_{t+1}}}. 
    \end{equation*}
Reordering the above inequality gives
    \begin{equation*}
        w(i_t,i_{t+1}) = \frac{v_{i_t}(A_{i_{t+1}})}{w_{i_{t+1}}} - \frac{v_{i_t}(A_{i_t})}{w_{i_t}} \leq \frac{p_{i_{t}}}{w_{i_t}} - \frac{p_{i_{t+1}}}{w_{i_{t+1}}}, 
    \end{equation*}
    which implies that the total weight of the cycle is non-positive:
    \begin{equation*}
        w(\mathcal{C}) = \sum_{t=0}^{k-1} w(i_t,i_{t+1}) \leq \sum_{t=0}^{k-1}\left( \frac{p_{i_{t}}}{w_{i_t}} - \frac{p_{i_{t+1}}}{w_{i_{t+1}}} \right) = 0.
        \qedhere
    \end{equation*}
\end{proof}
Note that Lemma~\ref{lemma:no-positive-weight-cycles} provides a way to check if a given allocation $\mathbf{A}$ is WEFable in polynomial time. 
This can be done by using the Floyd-Warshall algorithm\footnote{As introduced in~\cite{floyd1962algorithm}, it is an algorithm for finding the shortest paths in a directed weighted graph with positive or negative edge weights.} to check whether the corresponding weighted envy graph $G_{\mathbf{A}}$ has any positive-weight cycles.
Constructing the weighted envy graph takes $O(n^2)$ time\footnote{For general allocation problems, constructing the weighted envy graph takes $O(n^2+nm)$ time.} and the Floyd-Marshall algorithm runs in $O(n^3)$ time.

\begin{corollary}[Checking WEFable]
\label{corollary:polytime-checking-WEFable}
Given an allocation $\mathbf{A}$ for an instance $\mathcal{I}$, it takes $O(n^3)$ time to check whether it is WEFable.
\end{corollary}

Next we introduce another equivalent characterization of WEFable allocations, which is related a property which we call \emph{permutation resistant}.

\begin{definition}[Permutation Resistant]
\label{definition:Permutation-Resistant}
An allocation $\mathbf{A}$ is permutation resistant if for any permutation $\sigma: [n] \to [n]$, it holds that 
\begin{equation*}
    \sum_{i\in N} \frac{v_i(A_i)}{w_i} \geq \sum_{i\in N} \frac{v_i(A_{\sigma(i)})}{w_{\sigma(i)}}.
\end{equation*}
\end{definition}

Note that on the RHS of the equation, the denominator also changes depending on $\sigma$.
Therefore the house allocation $\mathbf{A}$ that maximizes the LHS of the above equation is not necessarily WEFable (as we will show in the next section).

\begin{lemma}\label{lemma:equivalence-permutation-max}
An allocation $\mathbf{A}$ is WEFable if and only if it is permutation resistant.
\end{lemma}
\begin{proof}
   We first prove the ``only if'' direction, which states that $\mathbf{A}$ is permutation resistant if $\mathbf{A}$ is WEFable.
    By Lemma~\ref{lemma:no-positive-weight-cycles}, $G_\mathbf{A}$ has no positive-weight cycles. Thus for any permutation $\sigma$, the edges $F = \{ (i,j) : j = \sigma(i) \text{ and } j\neq i \}$ form a collection of cycles. Thus
    \begin{equation*}
        \sum_{i\in N} \frac{v_i(A_{\sigma(i)})}{w_{\sigma(i)}} - \sum_{i\in N} \frac{v_i(A_i)}{w_i} = \sum_{(i,j)\in F} w(i,j) \leq 0,
    \end{equation*}
    which implies that the allocation $\mathbf{A}$ is permutation resistant .

    Next, we prove the ``if'' direction, which states that the allocation $\mathbf{A}$ is WEFable if it is permutation resistant.
    Consider any cycle $\mathcal{C} = (i_0,i_1,\ldots,i_k)$ in $G_\mathbf{A}$, where $i_0 = i_k$.
    Define a permutation $\sigma$, in which $\sigma(i_t) = i_{t+1}$ for all $t\in \{0,1,\ldots,k-1\}$ and $\sigma(i) = i$ for all $i\notin \{i_0,i_1,\ldots,i_{k-1}\}$.
    Since $\mathbf{A}$ is permutation resistant, we have
    \begin{equation*}
        w(\mathcal{C}) = \sum_{t=0}^{k-1} w(i_t,i_{t+1}) = \sum_{t=0}^{k-1} \left( \frac{v_{i_t}(A_{\sigma(i_t)})}{w_{\sigma(i_t)}} - \frac{v_{i_t}(A_{i_t})}{w_{i_t}} \right) = \sum_{i\in N} \frac{v_i(A_{\sigma(i)})}{w_{\sigma(i)}} - \sum_{i\in N} \frac{v_i(A_i)}{w_i} \leq 0.
    \end{equation*}
    
    Therefore, $G_\mathbf{A}$ has no positive-weight cycles, and thus $\mathbf{A}$ is WEFable.
\end{proof}

In summary, we have the following equivalent characterizations for WEFable allocations.

\begin{theorem}
\label{theorem:WEFable-properties}
 For an allocation $\mathbf{A}$, the following statements are equivalent.
 \begin{enumerate} [label=(\alph*), itemsep=5pt, leftmargin=20pt,itemindent=10pt]
     \item $\mathbf{A}$ is weighted envy-freeable.
     \item $G_\mathbf{A}$ has no positive-weight cycles.
    \item $\mathbf{A}$ is a permutation resistant.
 \end{enumerate}
\end{theorem}

We further determine the minimum required subsidy for a given WEFable allocation $\mathbf{A}$.
In fact, this payment vector is exactly the one we constructed in the proof of Lemma~\ref{lemma:no-positive-weight-cycles}.

\begin{theorem}
For a WEFable allocation $\mathbf{A}$, let $p_{i}^*=w_i\cdot \ell(i)$ for each $i \in N$,
where $\ell(i)$ is the maximum weight of any path starting at $i$ in $G_\mathbf{A}$.
Then, for every $\mathbf{P} \in \mathcal{P}(\mathbf{A})$ and $i \in N$, $p_i \geq p_{i}^*$.
\end{theorem}
\begin{proof}
We proof the theorem by contradiction: assume that there exists a subsidy $\mathbf{P}=(p_1,p_2,\cdots,p_n)$ such that for a fixed agent $i_0$, $p_{i_0} < p_{i_0}^*$. 
Consider the maximum-weight path starting at $i_0$ in $\mathbf{A}$.
Suppose it is $\mathcal{L}= (i_0,i_1,\ldots,i_{k})$.
By the definition of $\ell(i_0)$, we have $\ell(i_0)=\sum_{t=0}^{k-1}w(i_t,i_{t+1}).$
Since $(\mathbf{A}, \mathbf{P})$ is WEF, for each $t\in \{0,1,\ldots,k-1\}$ we have
 \begin{equation*}
    \frac{v_{i_t}(A_{i_t}) + p_{i_t}}{w_{i_t}} \geq \frac{v_{i_t}(A_{i_{t+1}}) + p_{i_{t+1}}}{w_{i_{t+1}}}.
\end{equation*}

Reordering the above inequality gives
\begin{equation*}
    \frac{p_{i_{t}}}{w_{i_t}} - \frac{p_{i_{t+1}}}{w_{i_{t+1}}} \ge \frac{v_{i_t}(A_{i_{t+1}})}{w_{i_{t+1}}} - \frac{v_{i_t}(A_{i_t})}{w_{i_t}}= 
     w(i_t,i_{t+1}).
\end{equation*}
    
Summing this over all $t\in \{0,1,\ldots,k-1\}$, we have
 \begin{equation*}
        \frac{p_{i_0}}{w_{i_0}} - \frac{p_{i_{k-1}}}{w_{i_{k-1}}} \ge 
        \ell(i_0) =
        \frac{p_{i_0}^*}{w_{i_0}}.
    \end{equation*}

Reordering the inequality above, we obtain
\begin{equation*}
   p_{i_0} \ge p_{i_0}^{*} +\frac{w_{i_0}}{w_{i_{k-1}}}\cdot p_{i_{k-1}} \geq  p_{i_0}^{*},
\end{equation*}
which is a contradiction with the assumption that $p_{i_0} < p_{i_0}^*$.
\end{proof}

Note that each $p_{i}^*$ can be computed in $O(n^3)$ time\footnote{For general allocation problems, each $p_{i}^*$ can be computed in $O(n^3+nm)$ time.} by the Floyd-Marshall algorithm.
Therefore, given any WEFable allocation, it takes polynomial time to compute the minimum envy-eliminating subsidies.

\subsection{Non-existence of WEFable Allocations}

Observation~\ref{observation:EFable exist} indicates that for the unweighted case, every house allocation instance admits envy-freeable allocations. 
Unfortunately, this does not hold in the weighted setting.
The following example illustrates that WEFable allocations do not necessarily exist, even when there are only two agents and two houses.

\begin{example}[The Hard Instance]
\label{example:Hard-Instance}
    Consider the following hard instance with two agents $\left\{1,2\right\}$ and two houses $\{h_1,h_2\}$.
    The agents' utility functions are shown in Table~\ref{tab1} for any $\epsilon < 1$:
    \begin{table}[htb]
    \centering
 \renewcommand{\arraystretch}{1.2}
 \setlength\tabcolsep{12pt}
        \begin{tabular}{c|c|c}
    	&  $h_1$ & $h_2$\\ \hline
    	agent 1   & $\epsilon$ & $\epsilon$   \\
    	agent 2   & $1$ & $1$ 
        \end{tabular}
         \vspace{3mm}
        \caption{Hard Instance 1: suppose $w_1 < w_2$.}
         \label{tab1}
    \end{table}
    
    For the problem of house allocation, we can assume w.l.o.g. that $h_i$ is allocated to agent $i$.
    However, the corresponding envy-graph has a positive-weight cycle because
    \begin{equation*}
        w(1,2)+w(2,1) = \left( \frac{\epsilon}{w_2} - \frac{\epsilon}{w_1} \right) + \left( \frac{1}{w_1} - \frac{1}{w_2} \right) = \left( \frac{1}{w_1} - \frac{1}{w_2} \right) \cdot (1-\epsilon) > 0.
    \end{equation*}
    
    The weighted envy graph constructed for the allocation is shown below and we can observe that the weights of the cycle in Fig.~\ref{figure:An Example of Weighted Envy Graph} are positive.
    Therefore, this example demonstrates that WEFable allocations do not always exist.
    Note that $\epsilon$ can be $0$, which means that even for binary instances with two agents, WEFable allocations do not always exist.
    
    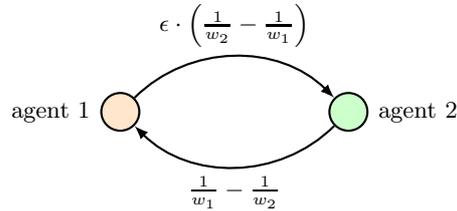
\begin{figure}[htb]
    \captionsetup{font=footnotesize}
    \centering
    \begin{tikzpicture}
     [
        node distance=3cm,
        on grid,
        thick,
        font=\small]
     
    \node
    [ 
        state,
        fill=orange!20,
        align=center,
        inner sep=1mm,
        minimum size=5mm,
    ] (q_0) [label=left: agent $1$]
    {
    };
     
    \node
    [
        state,
        fill=green!20,
        align=center,
        inner sep=5pt,
        minimum size=5mm,
    ] (q_1) [right=of q_0] [label=right:agent $2$]
    {
    
    };
     
    \path [-latex]
        (q_0) edge [bend left=45] 
            node [above] {$\epsilon \cdot \left(\frac{1}{w_2} - \frac{1}{w_1}\right)$} (q_1)
        (q_1) edge [bend left=45]   
            node [below] {$\frac{1}{w_1} - \frac{1}{w_2}$} (q_0);
    \end{tikzpicture}
    
    \captionsetup{font=small}
    \caption{The corresponding Weighted Envy Graph.}
    \label{figure:An Example of Weighted Envy Graph}
    \end{figure}
\end{example}

We can also verify whether the allocation if WEFable or not by the property of permutation resistant.
For the hard instance in Example~\ref{example:Hard-Instance}, any allocation has weighted social welfare $\frac{\epsilon}{w_1} + \frac{1}{w_2}$, which is smaller than $\frac{\epsilon}{w_2} + \frac{1}{w_1}$, i.e., its value under the permutation $\sigma:\sigma(1)=2; \sigma(2)=1$.

\paragraph{Remark.}
We remark that the non-existence of WEFable is due to the existence of agent weights and the constraint that each agent gets exactly one house.
If any of the two requirements is remove, WEFable allocations always exist.
As Observation~\ref{observation:EFable exist} shows, when agents have equal weights, a WEFable allocation can be easily found by computing a maximum weight matching.
Without the house allocation constraint, we can also get a WEFable allocation easily, by allocation all houses to the agent $i$ with maximum $v_i(H) = \sum_{h\in H} v_i(h)$ (here we assume the utility function is additive: it can be verified that this allocation is WEFable because it is permutation resistant), and setting $p_j = \frac{w_j}{w_i} \cdot v_j(H)$ for all other agent $j \in N \setminus \{i\}$.

\medskip

Note that the weighted envy-graph is not utility function scale-free because if we scale the utility function of agent $i$, only the outgoing edges of $i$ change their weight (same as \cite{SAGT2019,conf/sigecom/BrustleDNSV20,conf/wine/WuZZ23}). 
It is possible that an instance admits WEFable allocations under one normalization but does not under a different normalization.
In fact, if we scale the utility function of agent $1$ in Example~\ref{example:Hard-Instance} by a factor of $1/\epsilon$ (suppose $\epsilon>0$), then the instance admits WEFable allocations.
Nevertheless, the instance (and the weighted envy graph) is still scale-free if we scale all utility functions by the same ratio.
It is thus interesting to study whether WEFable allocations exist under some natural normalization assumptions.
As Example~\ref{example:Hard-Instance} shows, the standard normalization (used by \cite{SAGT2019,conf/sigecom/BrustleDNSV20,conf/wine/WuZZ23}) of assuming $v_i(h)\leq 1$ for all $i\in N$ and $h\in H$ does not help ensuring the existence of WEFable allocation.
We then turn to the other normalization assumption, assuming all agents have additive utility functions and $v_i(H) = 1$ for all $i\in N$.
Unfortunately, we show that WEFable allocations do not always exist under this assumption, unless $n=2$.

\begin{lemma} \label{lemma:normalized}
    For normalized instances ($v_i(H) = 1$ for all $i\in N$), WEFable allocations always exist when $n=2$. When $n\geq 3$, WEFable allocations do not always exist.
\end{lemma}
\begin{proof}
    When $n=2$, there must exist houses $h_1, h_2$ such that $v_1(h_1) \geq v_2(h_1)$ and $v_2(h_2) \geq v_1(h_2)$ (i.e., the allocation $(A_1 = h_1, A_2 = h_2)$ is unweighted envy-free). 
    Thus the allocation $\mathbf{A}$ is weighted envy-freeable since $G_\mathbf{A}$ does not have positive-weight cycles due to
    \begin{equation*}
        w(1,2)+w(2,1) = \frac{v_1(h_2) - v_2(h_2)}{w_2} + \frac{v_2(h_1) - v_1(h_1)}{w_1} \leq 0.
    \end{equation*}

    Unfortunately, weighted envy-freeable allocations are not guaranteed to exist when $n\geq 3$. 
    Consider the following example:
    
    \begin{itemize}
        \item[$\bullet$] let the weights of the $n$ agents satisfy $w_1 < w_2 < \cdots < w_n$;
        \item[$\bullet$] let $v_1(h_1) = v_1(h_2) = 0.5 - \epsilon$, $v_1(h_3) = 2\epsilon$, $v_1(h) = 0$ for all other houses $h$;
        \item[$\bullet$] let $v_2(h_1) = v_1(h_2) = 0.5$, $v_2(h) = 0$ for all other houses $h$;
        \item[$\bullet$] for all $i\geq 3$, let $v_i(h_i) = 1$; $v_i(h) = 0$ for all $h\neq h_i$.
    \end{itemize}
    
    \begin{table}[htb]
    \centering
     \renewcommand{\arraystretch}{1.2}
 \setlength\tabcolsep{12pt}
        \begin{tabular}{c|c|c|c|c|c}
    	&  $h_1$ & $h_2$ & $h_3$ & $\cdots$ & $h_n$\\ \hline
    	agent $1$   & $0.5-\epsilon$ & $0.5-\epsilon$ & $2\epsilon$ & $\cdots$ & $0$  \\
    	agent $2$   & $0.5$ & $0.5$  & $0$ & $\cdots$ & $0$ \\
    	agent $3$   & $0$ & $0$  & $1$ & $\cdots$ & $0$ \\
    	$\vdots$   & $\vdots$ & $\vdots$  & $\vdots$ & $\ddots$ & $\vdots$ \\
    	agent $n$   & $0$ & $0$  & $0$ & $\cdots$ & $1$ \\
        \end{tabular}
        \vspace{3mm}
        \caption{Hard Instance 2: suppose $w_1 < w_2 < \cdots < w_n$.}
    \end{table}

    For all $i\geq 3$, if house $h_i$ is not allocated to agent $i$, then agent $i$ will form a positive-weight cycle with the agent $j$ who receives $h_i$, because
    \begin{equation*}
        w(i,j) + w(j,i) = \frac{1}{w_j} + \frac{v_j(A_i)}{w_i} - \frac{v_j(h_i)}{w_j} \geq \frac{1}{w_j} - \frac{2\epsilon}{w_j} > 0.
    \end{equation*}
    
    However, if $h_i$ is allocated to agent $i$ for all $i\geq 3$, then there will be a positive-weight cycle between agent $1$ and $2$, in a way similar to Example~\ref{example:Hard-Instance}.
\end{proof}

\section{Weighted Envy-freeable Allocations in Tractable Cases}

Given the non-existence results we have presented, a natural open question is whether we can decide the existence of WEFable allocations (and output one if they exist) in polynomial time.
In this section, we present polynomial time algorithms for the problem in some natural special cases.

\subsection{Identical Utility Functions}
We first consider the case where all agents have the same utility function for the houses, that is, $v_i(h) = v_j(h)$ for all $i, j \in N$ and $h \in H$. 
We denote the common utility function by $v$.

\begin{theorem}
\label{theorem:Identical-is-WEFable}
    When agents have identical utility functions, all allocations are WEFable.
\end{theorem}
\begin{proof}
Consider any allocation $\mathbf{A}=(A_1,\dots,A_n)$ and for any permutation $\sigma: [n] \to [n]$, we have $$\sum_{i \in N}{\frac{v(A_i)}{w_i} = \sum_{i\in N}{\frac{v(A_{\sigma(i)})}{w_{\sigma(i)}}}}.$$
It means that the allocation $\mathbf{A}$ is WEFable by Theorem~\ref{theorem:WEFable-properties}.
Then, we only need to find the agent $i$ with maximum $\frac{v(A_i)}{w_i}$, and subsidize other agent $j \in N \setminus \{i\}$ with $p_j=\frac{w_j}{w_i}\cdot v(A_i) - v(A_j)$.
The outcome is WEF because (let $p_i = 0$):
\begin{equation*}
    \forall j \in N,\ \ \frac{v(A_j)+p_j}{w_j} = \frac{v(A_i)}{w_i}.
    \qedhere
\end{equation*}
\end{proof}

\subsection{Two Types of Agents}

Next we consider the case when there are only two types of agents, where agents of the same type have the same weight and utility function.
The setting captures our motivating example for weighted house allocation, in which the agents can be categorized into large families and small families.
Formally speaking, the set of agents is partitioned into large agents $N_L$ and small agents $N_S$, where all agents in $N_L$ have weight $w_L$ and utility function $v_L$; all agents in $N_S$ have weight $w_S$ and utility function $v_S$.
Let $n_L = |N_L|$ and $n_S = |N_S|$ be the number of large and small agents, respectively.
We have $n_L + n_S = n \leq m$.
Note that WEFable allocations are not guaranteed to exist even when there are only two types of agents, as Example~\ref{example:Hard-Instance} shows.

\begin{algorithm}[htp]
\caption{WEFable Allocations for Two Types of Agents}
\label{Alg:two-types-of-agents}
\KwIn{An instance $\mathcal{I}=(N_L\cup N_S,H,(w_L,w_S),(v_L,v_S))$ with two types of agents;}
\KwOut{A WEFable allocation $\mathbf{A}$ or ``No such allocation'';}
Sort and reindex the houses in descending order of $v_L(h)-v_S(h)$:
\begin{equation*}
    v_L(h_1)-v_S(h_1) \geq v_L(h_2)-v_S(h_2) \geq \cdots \geq v_L(h_m)-v_S(h_m).
\end{equation*}\\
\If{$\frac{1}{w_L}(v_L(h_{n_L}) - v_S(h_{n_L}))\geq \frac{1}{w_S}(v_L(h_{m-n_S+1}) - v_S(h_{m-n_S+1}))$}
{
Assign houses $\{ h_1,h_2,\ldots,h_{n_L} \}$ to agents in $N_L$ and houses $\{ h_{m-n_S+1},\ldots,h_{m} \}$ to agents in $N_S$, making sure that each agent receives one house\;
\textbf{return} the corresponding allocation $\mathbf{A}$;
}
\Else
{
\textbf{output} ``No such allocation''.
}
\end{algorithm}

We prove the correctness of our algorithm by showing that if the condition in line 2 is satisfied, then $G_\mathbf{A}$ of the returned allocation $\mathbf{A}$ does not contain any positive-weight cycle (which by Theorem~\ref{theorem:WEFable-properties} implies that the allocation is WEFable);  if the condition in line 2 is not satisfied, then no allocation is WEFable.
We first show that for any allocation $\mathbf{A}$, to check whether $G_\mathbf{A}$ contains positive-weight cycles, it suffices to check cycles of length $2$.

\begin{lemma} \label{lemma:length-2-cycle-suffices}
For any allocation $\mathbf{A}$ for an instance with two types of agents, $G_\mathbf{A}$ contains a positive-weight cycle if and only if it contains a positive-weight cycle with length $2$.
Moreover, the two agents in the length-$2$ cycle are of different types.
\end{lemma}
\begin{proof}
    Since the ``if'' direction is trivial, we only focus on the ``only if'' direction, i.e., if $G_\mathbf{A}$ contains a positive-weight cycle, there must exists a positive-weight cycle with length $2$, in which the two agents are of different types.

    We call an edge $(i,j)$ in $G_\mathbf{A}$ \emph{cross-type} if agents $i$ and $j$ have different types.
    Note that any positive-weight cycle $\mathcal{C} = (i_0,i_1,\ldots,i_k)$ (where $i_0=i_k$) must contain at least one cross-type edge, because otherwise the total weight of the cycle is $0$: (suppose all agents in the cycle are large)
    \begin{equation*}
        w(\mathcal{C}) = \sum_{t=0}^{k-1} w(i_t,i_{t+1}) = \sum_{t=0}^{k-1}\left( \frac{v_L(A_{i_{t+1}})}{w_{L}} - \frac{v_L(A_{i_{t}})}{w_{L}} \right) = 0.
    \end{equation*}
    
    Next we show that any positive-weight cycle can be ``compressed'' into a cycle containing only cross-type edges, while preserving the total weight of the cycle.
    Let $\mathcal{L} = (i_0,i_1,\ldots,i_k)$ be a path contained in the cycle, where $(i_{k-1},i_{k})$ is a cross-type edge and all other edges are not.
    We further assume that agents $i_0,\ldots,i_{k-1}$ are large and $i_k$ is small (the other case can be proved in a similar way).
    We show that the total weight of the path $\mathcal{L} = (i_0,i_1,\ldots,i_k)$ is equal to the path $(i_0, i_k)$ with a single cross-type edge (see Figure~\ref{figure:path-compression} for an example):
    
    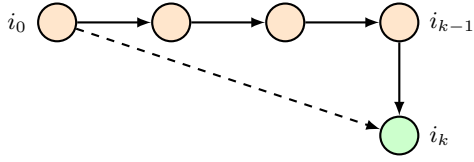
\begin{figure}[htb]
    \captionsetup{font=footnotesize}
    \centering
    \begin{tikzpicture}
     [
        node distance=1.5cm,
        on grid,
        thick,
        font=\small]
     
    \node
    [ 
        state,
        fill=orange!20,
        align=center,
        inner sep=1mm,
        minimum size=5mm,
    ] (q_0) [label=left: $i_0$]
    {
    };
     
    \node
    [
        state,
        fill=orange!20,
        align=center,
        inner sep=5pt,
        minimum size=5mm,
    ] (q_1) [right=of q_0] []
    {
    
    };
     
     \node
    [ 
        state,
        fill=orange!20,
        align=center,
        inner sep=1mm,
        minimum size=5mm,
    ] (q_2) [right=of q_1]
    {
    };
     
     \node
    [ 
        state,
        fill=orange!20,
        align=center,
        inner sep=1mm,
        minimum size=5mm,
    ] (q_3) [right=of q_2, label=right: $i_{k-1}$]
    {
    };
     
      \node
    [ 
        state,
        fill=green!20,
        align=center,
        inner sep=1mm,
        minimum size=5mm,
    ] (q_4) [below=of q_3, label=right: $i_{k}$]
    {
    };
    
    \path [-latex]
        (q_0) edge [bend left=0] 
            node [above] {} (q_1)
        (q_1) edge [bend left=0]   
            node [above] {} (q_2)
        (q_2) edge [bend left=0]   
            node [above] {} (q_3)
        (q_3) edge [bend left=0]   
            node [above] {} (q_4);
            
    \path [-latex, dashed]
      (q_0) edge [bend left=0]   
            node [above] {} (q_4); 
    \end{tikzpicture}
    \captionsetup{font=small}
    \caption{An example of path compression, where orange nodes represent the large agents and the green nodes represent the small agents.
    The edge $(i_{k-1}, i_k)$ and $(i_{0}, i_k)$ are the cross-type edges.}
    \label{figure:path-compression}
    \end{figure}
    
    \vspace{-15pt}
    
    \begin{align*}
        w(\mathcal{L}) & = \sum_{t=0}^{k-1} w(i_t,i_{t+1}) = 
        \sum_{t=0}^{k-2}\left( \frac{v_L(A_{i_{t+1}})}{w_{L}} - \frac{v_L(A_{i_{t}})}{w_{L}} \right) +
        \left(\frac{v_L(A_{i_{k}})}{w_{S}} - \frac{v_L(A_{i_{k-1}})}{w_{L}} \right)\\
        & = \left( \frac{v_L(A_{i_{k-1}})}{w_{L}} - \frac{v_L(A_{i_{0}})}{w_{L}} \right) +
        \left(\frac{v_L(A_{i_{k}})}{w_{S}} - \frac{v_L(A_{i_{k-1}})}{w_{L}} \right)
        = \frac{v_L(A_{i_{k}})}{w_{S}} - \frac{v_L(A_{i_{0}})}{w_{L}} = w(i_0,i_k).
    \end{align*}
    
    By repeatedly compressing paths as shown above, we obtain a cycle with only cross-type edges, whose total weight is the same as the original cycle.
    Let $\mathcal{C'} = (i_0,i_1,\ldots,i_k)$ be the resulting cycle, where $i_0 = i_k$, $k\geq 2$ is even, $i_0,i_2,\ldots,i_{k-2}$ are large agents and $i_1,i_3,\ldots,i_{k-1}$ are small agents.
    We further let $\overleftarrow{\mathcal{C}}$ be the cycle obtained by reversing the direction of all edges in $\mathcal{C'}$ (see Figure~\ref{figure:C-and-reverse-C} for an example).
    It can be verified that $w(\overleftarrow{\mathcal{C}}) = w(\mathcal{C'})$:
    \begin{align*}
        w(\mathcal{C'}) & = \sum_{\text{even } t<k} \left( \frac{v_L(A_{i_{t+1}})}{w_{S}} - \frac{v_L(A_{i_{t}})}{w_{L}} \right) +
        \sum_{\text{odd } t<k} \left( \frac{v_S(A_{i_{t+1}})}{w_{L}} - \frac{v_S(A_{i_{t}})}{w_{S}} \right) \\
        & = \sum_{\text{even } t<k} \left( \frac{v_S(A_{i_{t}})}{w_{L}} - \frac{v_L(A_{i_{t}})}{w_{L}} \right) +
        \sum_{\text{odd } t<k} \left( \frac{v_L(A_{i_{t}})}{w_{S}} - \frac{v_S(A_{i_{t}})}{w_{S}} \right) \\
        & = \sum_{\text{even } t<k} \left( \frac{v_S(A_{i_{t}})}{w_{L}} - \frac{v_S(A_{i_{t+1}})}{w_{S}} \right) +
        \sum_{\text{odd } t<k} \left( \frac{v_L(A_{i_{t}})}{w_{S}} - \frac{v_L(A_{i_{t+1}})}{w_{L}} \right) = w(\overleftarrow{\mathcal{C}}).
    \end{align*}
    
    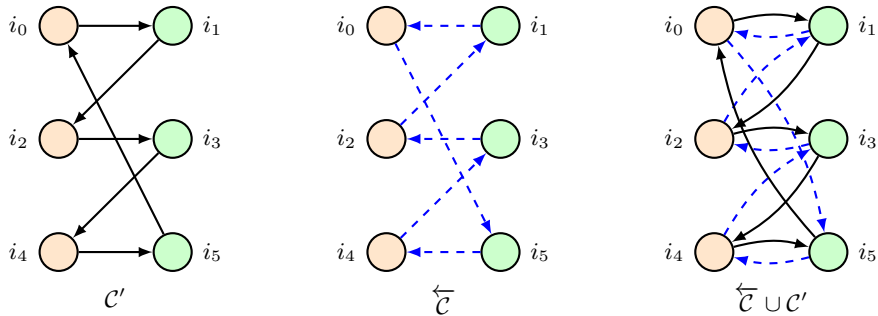
\begin{figure}[htb]
    \captionsetup{font=footnotesize}
    \centering
    \begin{tikzpicture}
     [
        node distance=1.5cm,
        on grid,
        thick,
        font=\small]
     
    \node
    [ 
        state,
        fill=orange!20,
        align=center,
        inner sep=1mm,
        minimum size=5mm,
    ] (q_0) [label=left: $i_0$]{};
     
    \node
    [
        state,
        fill=orange!20,
        align=center,
        inner sep=5pt,
        minimum size=5mm,
    ] (q_1) [below=of q_0, label=left: $i_2$] {};
     
     \node
    [ 
        state,
        fill=orange!20,
        align=center,
        inner sep=1mm,
        minimum size=5mm,
    ] (q_2) [below=of q_1, label=left: $i_4$]{};
     
     \node
    [ 
        state,
        fill=green!20,
        align=center,
        inner sep=1mm,
        minimum size=5mm,
    ] (q_3) [right=of q_0, label=right: $i_1$]{};
     
      \node
    [ 
        state,
        fill=green!20,
        align=center,
        inner sep=1mm,
        minimum size=5mm,
    ] (q_4) [below=of q_3, label=right: $i_3$]{};
    
      \node
    [ 
        state,
        fill=green!20,
        align=center,
        inner sep=1mm,
        minimum size=5mm,
    ] (q_5) [below=of q_4, label=right: $i_5$]{};

    \path [-latex]
        (q_0) edge [bend left=0] 
            node [above] {} (q_3)
        (q_3) edge [bend left=0]   
            node [above] {} (q_1)
        (q_1) edge [bend left=0]   
            node [above] {} (q_4)
        (q_4) edge [bend left=0]   
            node [above] {} (q_2)
        (q_2) edge [bend left=0]   
            node [below=1em] {$\mathcal{C}'$} (q_5)
        (q_5) edge [bend left=0]   
            node [above] {} (q_0)   
            ;
    
    \node
    [ 
        state,
        fill=orange!20,
        align=center,
        inner sep=1mm,
        minimum size=5mm,
    ] (q_6) [right=8em of q_3, label=left: $i_0$]{};
     
    \node
    [
        state,
        fill=orange!20,
        align=center,
        inner sep=5pt,
        minimum size=5mm,
    ] (q_7) [below=of q_6, label=left: $i_2$]{};
     
     \node
    [ 
        state,
        fill=orange!20,
        align=center,
        inner sep=1mm,
        minimum size=5mm,
    ] (q_8) [below=of q_7, label=left: $i_4$]{};
     
     \node
    [ 
        state,
        fill=green!20,
        align=center,
        inner sep=1mm,
        minimum size=5mm,
    ] (q_9) [right=of q_6, label=right: $i_1$]{};
     
      \node
    [ 
        state,
        fill=green!20,
        align=center,
        inner sep=1mm,
        minimum size=5mm,
    ] (q_10) [below=of q_9, label=right: $i_3$]{};
    
      \node
    [ 
        state,
        fill=green!20,
        align=center,
        inner sep=1mm,
        minimum size=5mm,
    ] (q_11) [below=of q_10, label=right: $i_5$]{};
     
    \path [-latex, dashed]
        (q_9) edge [bend left=0, color=blue] 
            node [above] {} (q_6)
        (q_7) edge [bend left=0, color=blue]   
            node [above] {} (q_9)
        (q_10) edge [bend left=0, color=blue]   
            node [above] {} (q_7)
        (q_8) edge [bend left=0, color=blue]   
            node [above] {} (q_10)
        (q_11) edge [bend left=0, color=blue]   
            node [below=1em,  color=black] {$\overleftarrow{\mathcal{C}}$} (q_8)
        (q_6) edge [bend left=0, color=blue]   
            node [above] {} (q_11)   
            ;
    
      \node
    [ 
        state,
        fill=orange!20,
        align=center,
        inner sep=1mm,
        minimum size=5mm,
    ] (q_12) [right=8em of q_9, label=left: $i_0$]{};
     
    \node
    [
        state,
        fill=orange!20,
        align=center,
        inner sep=5pt,
        minimum size=5mm,
    ] (q_13) [below=of q_12, label=left: $i_2$]{};
     
     \node
    [ 
        state,
        fill=orange!20,
        align=center,
        inner sep=1mm,
        minimum size=5mm,
    ] (q_14) [below=of q_13, label=left: $i_4$]{};
     
     \node
    [ 
        state,
        fill=green!20,
        align=center,
        inner sep=1mm,
        minimum size=5mm,
    ] (q_15) [right=of q_12, label=right: $i_1$]{};
     
      \node
    [ 
        state,
        fill=green!20,
        align=center,
        inner sep=1mm,
        minimum size=5mm,
    ] (q_16) [below=of q_15, label=right: $i_3$]{};
    
      \node
    [ 
        state,
        fill=green!20,
        align=center,
        inner sep=1mm,
        minimum size=5mm,
    ] (q_17) [below=of q_16, label=right: $i_5$]{};  
        
    \path [-latex, dashed]
        (q_15) edge [bend left=15, color=blue]
            node [above] {} (q_12)
        (q_13) edge [bend left=15, color=blue]   
            node [above] {} (q_15)
        (q_16) edge [bend left=15, color=blue]   
            node [above] {} (q_13)
        (q_14) edge [bend left=15, color=blue]   
            node [above] {} (q_16)
        (q_17) edge [bend left=15, color=blue]   
            node [below=.5em, color=black] {$\overleftarrow{\mathcal{C}}\cup \mathcal{C}'$} (q_14)
        (q_12) edge [bend left=15, color=blue]   
            node [above] {} (q_17)   
            ;
        
        \path [-latex]
        (q_12) edge [bend left=15] 
            node [above] {} (q_15)
        (q_15) edge [bend left=15]   
            node [above] {} (q_13)
        (q_13) edge [bend left=15]   
            node [above] {} (q_16)
        (q_16) edge [bend left=15]   
            node [above] {} (q_14)
        (q_14) edge [bend left=15]   
            node [below=1em] {} (q_17)
        (q_17) edge [bend left=15]   
            node [above] {} (q_12)   
            ;
    \end{tikzpicture}
    \captionsetup{font=small}
    \caption{Example of $\overleftarrow{\mathcal{C}}$ and $\mathcal{C}'$, where the black edges are from $\mathcal{C}'$ and blue edges are from $\overleftarrow{\mathcal{C}}$.}
    \label{figure:C-and-reverse-C}
    \end{figure}
    Therefore we have $w(\overleftarrow{\mathcal{C}}) > 0$.
    Since $\overleftarrow{\mathcal{C}}\cup \mathcal{C}'$ can be regarded as a collection of length-$2$ cycles, each of which has one large agent and one small agent, at least one of them has positive weight, as claimed in the lemma.
\end{proof}

Given the above lemma, to compute an allocation that is WEFable, it suffices to make sure that there is no length-$2$ positive-weight cycle consists of a large agent and a small agent in the corresponding weighted envy graph.

\begin{lemma} \label{lemma:two-types-agents-Yes-case}
If the condition in line 2 of Algorithm~\ref{Alg:two-types-of-agents} is satisfied, then the allocation $\mathbf{A}$ returned by the algorithm is WEFable.
\end{lemma}
\begin{proof}
    Fix an arbitrary large agent $i\in N_L$ and small agent $j\in N_S$.
    We show that the cycle $(i,j,i)$ has non-positive weight:
    \begin{align*}
        w(i,j) + w(j,i) & = \frac{v_L(A_j)}{w_S} - \frac{v_L(A_i)}{w_L}
        + \frac{v_S(A_i)}{w_L} - \frac{v_S(A_j)}{w_S} \\
        & = \frac{v_L(A_j) - v_S(A_j)}{w_S} - \frac{v_L(A_i) - v_S(A_i)}{w_L} \\
        & \leq \frac{v_L(h_{m-n_S+1}) - v_S(h_{m-n_S+1})}{w_S} - \frac{v_L(h_{n_L}) - v_S(h_{n_L})}{w_L} \leq 0,
    \end{align*}
    where the first inequality holds because $h_{m-n_S+1}$ is the house with maximum value of $v_L(h)-v_S(h)$ that is allocated to small agents and $h_{n_L}$ is the house with minimum value of $v_L(h)-v_S(h)$ that is allocated to large agents; the last inequality holds by the ``If'' condition in line 2 of Algorithm~\ref{Alg:two-types-of-agents}.
    
    Since $i$ and $j$ are fixed arbitrarily, we have shown that there is no length-$2$ positive-weight cycle consists of a large agent and a small agent in $G_\mathbf{A}$.
    Then by Lemma~\ref{lemma:length-2-cycle-suffices}, we conclude that $G_\mathbf{A}$ has no positive-weight cycle, and thus $\mathbf{A}$ is WEFable.
\end{proof}

Next we show that if Algorithm~\ref{Alg:two-types-of-agents} does not output any allocation, then the instance does not admit any WEFable allocations.

\begin{lemma} \label{lemma:two-types-agents-NO-case}
If the condition in line 2 of Algorithm~\ref{Alg:two-types-of-agents} is not satisfied, then all allocations for the instance are not WEFable.
\end{lemma}
\begin{proof}
    Let $\alpha = v_L(h_{n_L}) - v_S(h_{n_L})$ and $\beta = v_L(h_{m-n_S+1}) - v_S(h_{m-n_S+1})$. By assumption we have ${\alpha}/{w_L} < {\beta}/{w_S}$.
    By definition the number of houses $h$ with $v_L(h) - v_S(h) > \alpha$ is at most $n_L-1$; the number of houses with $v_L(h) - v_S(h) < \beta$ is at most $n_S-1$.
    
    Consider any allocation $\mathbf{A}$.
    Since there are $n_L$ large agents and $n_S$ small agents, there must exist $i\in N_L$ with $v_L(A_i) - v_S(A_i) \leq \alpha$ and $j\in N_S$ with $v_L(A_j) - v_S(A_j) \geq \beta$.
    Then we have
    \begin{equation*}
        w(i,j) + w(j,i) = \frac{v_L(A_j)}{w_S} - \frac{v_L(A_i)}{w_L}
        + \frac{v_S(A_i)}{w_L} - \frac{v_S(A_j)}{w_S} \geq \frac{\beta}{w_S} - \frac{\alpha}{w_L} > 0,
    \end{equation*}
    which implies that $\mathbf{A}$ is not WEFable and the instance does not admit any WEFable allocations.
\end{proof}

Putting the lemmas together, we prove the following.

\begin{theorem} \label{theorem:decide-WEFable-two-types}
Algorithm~\ref{Alg:two-types-of-agents} is a polynomial time algorithm that decides whether weighted envy-freeable allocations exist and, if so, computes one such allocation.
\end{theorem}
\begin{proof}
    By Lemma~\ref{lemma:two-types-agents-Yes-case} and~\ref{lemma:two-types-agents-NO-case}, Algorithm~\ref{Alg:two-types-of-agents} outputs a WEFable allocation if and only if the instance admits WEFable allocations. Thus the correctness of the algorithm follows.
    The core of the algorithm relies on sorting all houses based on $v_L(h) - v_S(h)$, which can be done in $O(m\log m)$ time.
    The other parts of the algorithm can be trivially done in $O(n)$ time, which completes the analysis for the polynomial running time, and the theorem follows.
\end{proof}

\subsection{Bi-valued Utility Functions}

Finally, we consider the bi-valued instances, in which we have $v_i(h) \in \{\epsilon,1\}$ for all $i\in N$ and $h\in H$, where $\epsilon \in [0,1)$.
We further assume that $m=n$.
Note that even for bi-valued instances, weighted envy-freeable allocations are not guaranteed to exist, e.g., see Example~\ref{example:Hard-Instance}.
Given a bi-valued instance $\mathcal{I}$, it is often more convenient to work with the graph $G_\mathcal{I}=(N,H,E)$, where we put an edge $(i,h)\in E$ between agent $i$ and house $h$ if and only if $v_i(h) = 1$.
We call the graph $G_\mathcal{I}$ a representing graph of the instance $\mathcal{I}$.
Every allocation for the instance can be represented by a matching in the representing graph, where the matched agents receive houses of value $1$ and each unmatched agent receives an arbitrary unique unmatched house (which has value $\epsilon$ to the agent).
Throughout, we call an agent \emph{matched} under allocation $\mathbf{A}$ if $v_i(A_i)=1$.

\begin{lemma}
\label{lemma:max-SW=PO}
For bi-valued instances, the allocation $\mathbf{A}$ is Pareto optimal if and only if the allocation $\mathbf{A}$ maximizes the number of matched agents.
\end{lemma}
\begin{proof}
The ``if'' direction is trivial because any Pareto improvement of the allocation $\mathbf{A}$ gives an allocation with a strictly larger number of matched agents.

Next we show the ``only if'' direction, i.e., any PO allocation $\mathbf{A}$ maximizes the number of matched agents.
Suppose otherwise, i.e., there exists an allocation $\mathbf{A}$ that is PO but the corresponding matching $\mathcal{M}$ is not maximum.
Berge's theorem~\cite{berge1957two} implies that there will be an augmenting path in $\mathcal{M}$, applying which to $\mathcal{M}$ results in a new matching $\mathcal{M}'$, in which all agents that are matched in $\mathcal{M}$ remain matched, and some unmatched agent in $\mathcal{M}$ becomes matched in $\mathcal{M'}$.
Therefore, the allocation corresponding to the augmented matching $\mathcal{M'}$ Pareto dominates $\mathbf{A}$, which contradicts the fact that the allocation $\mathbf{A}$ is PO.
\end{proof}

The above lemma implies that every PO allocation corresponds to a maximum matching in the representing graph.
Next we show that every WEFable allocation must be PO.

\begin{lemma} \label{lemma:WEFable-implies-PO}
For bi-valued instances with $m=n$, every WEFable allocation is Pareto optimal.
\end{lemma}
\begin{proof}
We prove the lemma by contradiction.
Suppose that there exists an allocation $\mathbf{A}$ that is WEFable but not PO.
By Lemma~\ref{lemma:max-SW=PO}, the matching corresponding to $\mathbf{A}$ is not maximum, and thus admits an augmenting path, applying which on $\mathbf{A}$ gives a new allocation $\mathbf{B}$.

Let $(i_0, i_1, \ldots, i_k)$ represent the augmenting path, where all agents $\{i_1,\ldots,i_k\}$ are matched in both $\mathbf{A}$ and $\mathbf{B}$, while agent $i_0$ is matched in $\mathbf{B}$ but not in $\mathbf{A}$.
The agents form a path in the sense that in allocation $\mathbf{B}$, agent $i_0$ claims the house $A_{i_1}$ that was allocated to $i_1$; agent $i_1$ claims house $A_{i_2}$, etc.
Since $m=n$, the house $i_k$ claims in $\mathbf{B}$ must be allocated to some agent (which we call $i_{k+1}$) in $\mathbf{A}$.
Moreover we must have $v_{i_{k+1}}(A_{i_{k+1}}) = 0$ as otherwise it is not an augmenting path.
In summary, for all $t=0,1,\ldots,k-1$, agent $i_t$ receives house $B_{i_t} = A_{i_{t+1}}$, where $v_{i_t}(A_{i_{t+1}}) = v_{i_{t+1}}(A_{i_{t+1}}) = 1$, while agent $i_k$ receives house $B_{i_k} = A_{i_{k+1}}$, where $v_{i_k}(A_{i_{k+1}}) = 1$ but $v_{i_{k+1}}(A_{i_{k+1}}) = 0$ (see Figure~\ref{figure:augmenting-path} for an example with $k=2$).

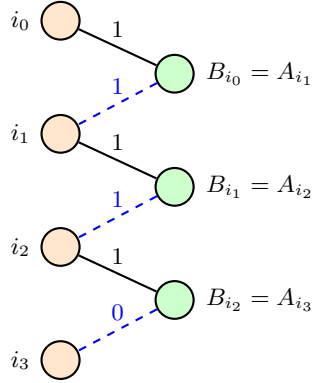
\begin{figure}[htb]
    \captionsetup{font=footnotesize}
    \centering
    \begin{tikzpicture}
     [
        node distance=1.5cm,
        on grid,
        thick,
        font=\small]
     
    \node
    [ 
        state,
        fill=orange!20,
        align=center,
        inner sep=1mm,
        minimum size=5mm,
    ] (q_0) [label=left: $i_0$]{};
     
    \node
    [
        state,
        fill=orange!20,
        align=center,
        inner sep=5pt,
        minimum size=5mm,
    ] (q_1) [below=of q_0, label=left: $i_1$] {};
     
     \node
    [ 
        state,
        fill=orange!20,
        align=center,
        inner sep=1mm,
        minimum size=5mm,
    ] (q_2) [below=of q_1, label=left: $i_2$]{};
     
     \node 
    [ 
        state,
        fill=green!20,
        align=center,
        inner sep=1mm,
        minimum size=5mm,
   ] (q_3) [below=2em of q_0, right of=q_0 ]{};
     \node [right=3.2em of q_3]{$B_{i_0}=A_{i_1}$};
    
      \node
    [ 
        state,
        fill=green!20,
        align=center,
        inner sep=1mm,
        minimum size=5mm,
    ] (q_4) [below=of q_3]{};
    \node [right=3.2em of q_4]{$B_{i_1}=A_{i_2}$};
    
    \node
    [ 
        state,
        fill=green!20,
        align=center,
        inner sep=1mm,
        minimum size=5mm,
    ] (q_5) [below=of q_4]{};
    \node [right=3.2em of q_5]{$B_{i_2}=A_{i_3}$};
    
    \node
    [ 
        state,
        fill=orange!20,
        align=center,
        inner sep=1mm,
        minimum size=5mm,
    ] (q_6) [below=of q_2, label=left: $i_3$]{};
        
    \path [-, dashed]
        (q_3) edge [bend left=0, color=blue]
            node [above] {$1$} (q_1)
        (q_4) edge [bend left=0, color=blue]   
            node [above] {$1$} (q_2)
        (q_5) edge [bend left=0, color=blue]   
            node [above] {$0$}(q_6)
            ;
        \path [-]
        (q_3) edge [bend left=0] 
            node [above] {$1$} (q_0)
        (q_4) edge [bend left=0]   
            node [above] {$1$} (q_1)
        (q_5) edge [bend left=0]   
            node [above] {$1$} (q_2)
            ;
    \end{tikzpicture}
    \captionsetup{font=small}
    \caption{Example of augmenting path with $k=2$, where black edges represent assignments in $\mathbf{B}$, blue edges represent assignments in $\mathbf{A}$. Orange nodes present agents and green nodes represent houses.}
    \label{figure:augmenting-path}
    \end{figure}

Define a permutation $\sigma$ as follows:
for all $t\in \{0,1,\ldots,k\}$, let $\sigma(i_t) = i_{t+1}$; let $\sigma(i_{k+1}) = i_0$; let $\sigma(i) = i$ for all $i\notin \{i_0,i_1,\ldots,i_{k+1}\}$.
Then we have
\begin{align*}
\sum_{i\in N}\frac{v_{i}(A_{\sigma(i)})}{w_{\sigma(i)}}
& = \sum_{i\notin \{i_0,\ldots,i_{k+1}\}}\frac{v_{i}(A_i)}{w_i}
+ \sum_{t=0}^k \frac{v_{i_t}(A_{i_{t+1}})}{w_{i_{t+1}}}
+ \frac{v_{i_{k+1}}(A_{i_0})}{w_{i_0}} \\
& = \sum_{i\notin \{i_0,\ldots,i_{k+1}\}}\frac{v_{i}(A_i)}{w_i}
+ \sum_{t=0}^k \frac{1}{w_{i_{t+1}}}
+ \frac{0}{w_{i_0}} > \sum_{i\in N}\frac{v_i(A_i)}{w_i},
\end{align*}
which contradicts the fact that $\mathbf{A}$ is WEFable.
\end{proof}

Combing the two lemmas above, we immediately have the following, because to check whether WEFable allocations exist, it suffices to check all PO allocations, each of which corresponds to a maximum matching.

\begin{theorem}
\label{theorem:whether-WEFable-exist}
    For any bi-valued instance with $m=n$, if the number of maximum matchings in the representing graph is polynomial (in $n$), then we have a polynomial time algorithm to check whether WEFable allocations exist.
\end{theorem}

\section{Conclusion and Open Problems}
Weighted fair division represents a substantial expansion of the fundamental fair division framework and has attracted considerable attention recently.
We are the first to consider weighted envy-freeness in house allocation problems.
We show that the existence of weighted envy-free house allocations can be efficiently determined by giving a polynomial time algorithm.
Furthermore, we provide a precise characterization of weighted envy-freeable allocations, show that weighted envy-freeable allocations may not always exist even with $2$ agents, and provide polynomial time algorithms for deciding their existence in some special cases.
Our work raises several intriguing open questions. 
A primary question is whether it is possible to determine, in polynomial time, if weighted envy-freeable allocations exist for a given instance in general.
A natural follow-up direction include investigating the amount of subsidy that is necessary to achieve weighted envy-freeness, if weighted envy-freeable allocations exist.
Another potential direction is to investigate the effects of strategic behavior in weighted house allocation problems.

%
%
%
 \bibliographystyle{splncs04}
 \bibliography{HAP}

\end{document}